\documentclass{article}

\author{}

\usepackage{amsmath}
\usepackage{amsthm}
\usepackage{amsfonts}
\usepackage{amssymb}
\usepackage{latexsym}
\usepackage{setspace}
\usepackage[dvips]{graphicx}
\usepackage{pstricks}
\usepackage{epsfig}

\newtheorem{theorem}{Theorem}[section]

\newtheorem{lemma}{Lemma}[section]
\newtheorem{proposition}{Proposition}[section]
\newtheorem{definition}[theorem]{Definition}
\newtheorem{example}{Example}[section]

\newtheorem{remark}{Remark}[section]

\begin{document}
\title{Archimedean-based Marshall-Olkin Distributions and Related Copula Functions}
\author{Sabrina
Mulinacci\thanks{University of Bologna, Department of Statistics, Via delle Belle Arti 41, 40126 Bologna, Italy. Phone: +(39)
0512094368; Fax: +(39)0512086242. e-mail:
sabrina.mulinacci@unibo.it}
}
\date{}
\maketitle
\begin{abstract}
 A new class of bivariate distributions is introduced that extends the Generalized Marshall-Olkin 
distributions of Li and Pellerey (2011). Their dependence structure is studied through the analysis of the copula functions that they induce.
These copulas, that include as special cases the Generalized Marshall-Olkin copulas and the Scale Mixture of Marshall-Olkin copulas (see Li, 2009),
 are obtained through suitable distortions of bivariate 
Archimedean copulas: this induces asymmetry, and 
the corresponding Kendall's tau as well as the tail dependence parameters are studied.
\end{abstract}
\bigskip


\section{Introduction} 
In their seminal paper (see Marshall and Olkin, 1967), the authors introduce the Marshall-Olkin distribution whose survival version is
$$\bar F(x_1,x_2)=\exp\{-\lambda _1x_1-\lambda _2x_2-\lambda _3\max (x_1,x_2) \}$$
$x_1,x_2\geq 0,\lambda_1,\lambda _2,\lambda _3> 0$. This is the distribution of a bivariate random vector $(M_1,M_2)$ with $M_1=\min (X_1,X_3)$ and $M_2=\min (X_2,X_3)$, where $X_1$, $X_2$, $X_3$ are three independent and exponentially distributed random variables.
\\
For an interpretation of the Marshall-Olkin distribution, 
consider a system with two components $C_1$ and $C_2$ (electronic elements, engines, credit obligors, life-insured 
married couples, etc.). The random variables $X_1$ and $X_2$ represent the arrival times of a shock causing 
the failure or default or death (depending on the case) of $C_1$ and $C_2$, respectively (idiosyncratic shocks), while 
$X_3$ is the arrival time of 
a shock causing the simultaneous end of the lifetimes of both $C_1$ and $C_2$ (systemick shock).
\medskip

In order to generalize this model, the case in which the underlying random variables $X_1,X_2,X_3$ 
are assumed to have marginal distributions different from the exponential one has been considered 
in the existing literature. The most general results, in this direction, are those obtained in Li and Pellerey (2011), 
where no restriction is made on the marginal distributions of the 
random variables $X_1,X_2,X_3$ while joint independence is again assumed. 
More precisely, the authors assume for $(X_1,X_2,X_3)$ a joint survival
distribution of type
\begin{equation}\label{lp}
\bar F(x,y,z)=\exp (-(H_1(x)+H_2(y)+H_3(z)))
\end{equation}
where the $H_i$ are the cumulative hazard functions of the random variables $X_i$.
The corresponding joint survival distribution of
the random variables $M_1$ and $M_2$ is
$$\bar F(x_1,x_2)=\exp\{-H_1(x_1)-H_2(x_2)-H_3(\max (x_1,x_2))\} $$
that is called \emph{Generalized Marshall-Olkin distribution}. 
These distributions incorporate, as special cases, the Marshall-Olkin type distributions introduced in Muliere and Scarsini 
(1987), the bivariate Weibull distributions introduced in Lu (1989) and the bivariate Pareto distributions introduced in Asimit et al. (2010). 
The authors analyze the dependence structure implied by these distributions introducing the 
corresponding 
copula functions, that they call \emph{Generalized Marshall-Olkin copulas}.
\bigskip

However, because of the independence among the original shocks (systemic and idiosyncratic), the dependence between  $C_1$ and $C_2$ lifetimes is only given by the occurrence of the systemic shock.
None the less, one could imagine concrete situations in which there is some additional dependence between $C_1$ and $C_2$ lifetimes given, for example, by some unobserved factors affecting all the original shocks (idiosyncratic and systemic) arrival times.\smallskip 

Consider for example the case of a firm interested in protecting itself against the failure of a production chain composed by a primary machine (or electronic component) 
in serial connection with two secondary machines in such a way  that the failure of the primary machine causes the failure of all the productive system, while the failure of a secondary machine causes only the failure of the corresponding production line.
Clearly the disease of all the three machines (representing the systemic and the idiosyncratic components of the system) is influenced by factors such as maintenance, good electric supply, and so on, and 
this fact induces dependence among them so that the disease of a production line can induce a change in the probability of failure of the primary machine.
An insurance policy written in order to protect against lack in the production caused by the failure of all,
or part, of the system, has to take into account the probability of failure of each machine (systemic and idiosyncratic
effects) and the dependence among them.\smallskip

The above is a lucky case, in which the systemic and the idiosyncratic components are well identified and observable.
This is not in general the case. Consider, for example, a life insurance contract written by a married couple 
with a final payment or goods transfer made when both spouses have died: such contracts are thought in order to transfer money or goods to heirs or to guarantee some
annuity until death, using, for example, the house property as collateral security (this is the case of reverse mortgages contracts). In such cases the contract expires when both spouses have died and their death can occur separately or simultaneously. Their simultaneous death can be caused by the occurrence of some systemic event (for example some ``catastrophic'' event like a car accident, an earthquake and so on) 
but also by something more related to some dependence among their lifetimes idiosyncratic  components: it is 
in fact well known that in very old people the death of one of the two spouses can induce in a very short period 
(so that the two events can be considered simultaneous)
the death of the other. However, notice that, some catastrophic events causing the simultaneous death of the married couple, 
cannot be considered as fully external and independent: think of the case of a simultaneous death caused by a car accident in consequence of the fact that the one of the two that was 
driving had an heart attack.
Clearly there are factors affecting both the simultaneous and the separated deaths: health care, 
social and government assistance, affective relationships and so on.\medskip

The idea of a common factor affecting all the components of a random vector is 
at the basis of the dependence structure represented by Archimedean copulas.
Such a dependence can be in fact obtained starting from a vector of i.i.d. exponentially distributed components 
and dividing all of them by a positive random variable representing a common stochastic intensity
(see Marshall and Olkin, 1988).
This is the case of Archimedean copulas with completely monotone generator. The case of $k$-dimensional Archimedean copulas with $k$-monotone generator can be similarly 
constructed starting from a vector uniformly distributed on the unit $k$-dimensional simplex and multiplying each component by a non-negative 
random variable, representing again a common factor affecting all the random vector (see McNeil and Ne\v{s}lehov\'{a}, 2009). 
In any case, all the random variables involved are modified by the common factor in the same way and, as it is 
well known, the induced dependence is symmetric.
\bigskip

In this paper we generalize the Li and Pellerey (2011) setting allowing for a 
dependence structure of Archimedean type among the original shocks arrival times. 
The restrictive fact of considering an exchangeable
dependence structure represents a first step in the perspective to include dependence: this 
is a compromise between mathematical tractability and realism of the assumptions.\medskip 

More precisely, we assume for $(X_1,X_2,X_3)$ a joint survival distribution 
function more general than that in (\ref{lp}), that is
$$\bar F(x,y,z)=G(H_1(x)+H_2(y)+H_3(z))$$
where $G:[0,+\infty)\rightarrow [0,1]$ is the generator of a three-variate Archimedean copula: in this case the associated copula is Archimedean with generator $G$ and $H_i(x)=G^{-1}(F_i(x))$ where 
$F_i$ the marginal distribution of $X_i$. The corresponding survival distribution of the random vector $(M_1,M_2)$ is
$$\bar F(x_1,x_2)=G(H_1(x_1)+H_2(x_2)+H_3(\max(x_1,x_2))$$
and we call it \emph{Archimedean-based Marshall-Olkin distribution} 
and the copula that it induces \emph{Archimedean-based Marshall-Olkin copula}. \par
As we will see, the family of the copulas so generated contains the two-dimensional Archimedean copulas,  the Generalized Marshall-Olkin copulas and other well known families as specific
 cases.
More precisely, we will show that any Archimedean-based Marshall-Olkin copula can be obtained through a suitable, in general asymmetric,
distortion of a bivariate Archimedean copula with generator $G$.
In the case of symmetric distortions, we recover, even if under more restrictive assumptions, the 
generalization of Archimedean copulas introduced in Durante et al. (2007).
In the case of linear distortions, we recover a proper subset of the Archimax copulas
 of Cap\'{e}ra\`{a} et al. (2000) that, when $G$ is the Laplace transform of a positive random variable, coincide with the 
bivariate Scale Mixture of Marshall-Olkin copulas studied and applied in Li (2009), Bernhart et al. (2013) and
Mai et al. (2013).\medskip

The impact of the distortion on the dependence structure induced by an Archimedean-based Marshall-Olkin copula 
is analyzed by studying its Kendall's function, its Kendall's tau and its tail dependence parameters.
\medskip

The paper is organized as follows. 
In Section \ref{AMO} we introduce  the Archimedean-based Marshall-Olkin distribution. 
In Section \ref{AMOCopula} we derive the Archimedean-based Marshall-Olkin copula. 
Section \ref{dp} is devoted to the analysis of some dependence properties of the copulas introduced: 
in particular, the expression of the Kendall's function and of the Kendall's tau are derived and, in some particular cases,
the tail dependence parameters are calculated.

\section{The Archimedean-based Marshall-Olkin Distribution}\label{AMO}

Let $(\Omega,\mathcal F,\mathbb P)$ be a probability space and $(X_1,X_2,X_3)$ be a random vector such that 
$\mathbb P(X_i>0)=1$ for every $i=1,2,3$. We assume that their joint survival distribution 
is of type
$$\bar F(x,y,z)=G(H_1(x)+H_2(y)+H_3(z))$$
where
\begin{itemize}
 \item $G:[0,+\infty)\rightarrow [0,1]$, with $G(0)=1$ and
$G^\prime$ exists on $(0,+\infty)$,
 it is non-positive, non-decreasing and concave; if $x_G=\inf\{x\geq 0:G(x)=0\}$, we have that $G$ is strictly decreasing on $[0,x_G)$ and so
its inverse function $G^{-1}:(0,1]\rightarrow [0,x_G)$ is well defined: we can extend $G^{-1}$ to the whole interval
$[0,1]$ by setting $G^{-1}(0)=x_G$. 
\item for every $i=1,2,3$, $H_i$ is continuous and strictly increasing for $x>0$, $H_i(x)=0$ for $x\leq 0$
and $\underset{x\rightarrow +\infty}\lim H_i(x)=x_G$.
\end{itemize}
Under these assumptions,
the joint survival distribution can be rewritten as
\begin{equation}\label{startdistr}\bar F(x,y,z)=G(G^{-1}(\bar F_1(x))+G^{-1}(\bar F_2(y))+G^{-1}(\bar F_3(z)))\end{equation}
with $\bar F_i=G\circ H_i$.  Thanks to Sklar's theorem and  Theorem 2 in McNeil and Ne\v{s}lehov\'{a} (2009), the $\bar F_i$, $i=1,2,3$, are the marginal survival distributions of the random variables 
$X_i$ and the dependence structure of the random vector $(X_1,X_2,X_3)$ is of Archimedean type with generator $G$.\medskip

Clearly, this setting includes, as a specific case ($G(x)=e^{-x}$), the family of the Generalized Marshall-Olkin distributions introduced in Li and Pellerey (2011).
\bigskip

In what follows, with abuse of notation, we will denote with $f(+\infty)$ 
the $\underset{x\rightarrow +\infty}\lim f(x)$, when it exists.
\begin{remark}\label{scarsini}
In Muliere and Scarsini (1987), the particular case with $G(x)=e^{-x}$ and $H_i(x)=\lambda _iH(x)$, $\lambda _i>0$, is studied and justified.

The choice of functions of type $H_i(x)=\lambda_iH(x)$ is allowed for any $G$ provided that 
$x_G=+\infty$. If $x_G<+\infty$, the requirement $ H_i(+\infty )=x_G$ is satisfied
if and only if $\lambda_1=\lambda_2=\lambda _3$ which corresponds to the case in which the vector $(X_1,X_2,X_3)$ is exchangeable.
\end{remark}
\bigskip
 
Let us now consider the random vector $(M_1,M_2)$ defined by $M_1=\min (X_1,X_3)$ and $M_2=\min (X_2,X_3)$. The corresponding survival distribution function, for $t_1,t_2>0$, is 
\begin{equation}\label{surv1}\begin{aligned}\bar F_{M_1,M_2}(t_1,t_2)&=\mathbb P(M_1>t_1,M_2>t_2)=\\
   &=\mathbb P(X_1>t_1,X_2>t_2,X_3>\max (t_1,t_2))=\\
&=G(H_1(t_1)+H_2(t_2)+H_3(\max(t_1,t_2))).
  \end{aligned}\end{equation}
We call the above distribution \emph{Archimedean-based Marshall-Olkin distribution}.\medskip

Setting $K_i(t_i)=H_i(t_i)+H_3(t_i)$, the marginal survival distributions are
\begin{equation}\label{margin}\bar F_{M_i}(t_i)=G(K_i(t_i)),\, i=1,2.\end{equation}
\begin{remark}
Notice that if $x_G<+\infty$, then $\bar F_{M_1,M_2}(t_1,t_2)=0$ on $\{(t_1,t_2):H_1(t_1)+H_2(t_2)+H_3(\max(t_1,t_2))\geq x_G\}$ and
$\bar F_{M_i}(t_i)=0$ for all $t_i\geq K_i^{-1}(x_G)$.\\
\end{remark}
\bigskip

\subsection{The singular component}
Considering insurance contracts based on the lifetimes $M_1$ and $M_2$, it is of some relevance 
to separate the impact on the joint distribution (and, equivalently, on the price)
of simultaneous failure or default or death (depending on the application in analysis) from the separated ones.

The simultaneous failure is clearly identified by the singular component of the distribution.
It is in fact well known that Marshall-Olkin distributions (as well as the Generalized Marshall-Olkin distributions introduced 
in Li and Pellerey, 2011) admit a singularity: clearly this fact continues to hold in our extended setting.\bigskip
\medskip

To simplify the notation,
we set
$$\hat H(t)=\sum_{i=1}^3H_i(t).$$
If $M=\min(X_1,X_2,X_3)$, its survival distribution function is given, for $t>0$,  by 
$$\mathbb P(M>t)=\mathbb P\left (X_1>t,X_2>t,X_3>t\right)=G(\hat H(t)).$$
\begin{proposition}\label{singularity}
Assume that $G$ is twice differentiable and that each $H_i$ is differentiable on $(0,+\infty)$. Then
$$\begin{aligned}
  F_{M_1,M_2}(t)&=\mathbb P(M_1\leq t,M_2\leq t)=\\
  &=F^a_{M_1,M_2}(t)+F^s_{M_1,M_2}(t)
  \end{aligned}
$$
where
\begin{equation}\label{price}F_{M_1,M_2}(t)=1+G(\hat H(t))-G(K_1(t))-G(K_2(t))\end{equation}
and
$$F^s_{M_1,M_2}(t)=-\int_0^tH_3^\prime(x)G^\prime (\hat H(x))dx.$$
\end{proposition}
\begin{proof}
(\ref{price}) is trivial.\\
Since
$$\frac{\partial^2}{\partial t_2\partial t_1}\bar F_{AMO}(t_1,t_2)=\left\{\begin{array}{c}
                           G^{\prime\prime}\left (H_1(t_1)+K_2(t_2)\right )H_1^\prime (t_1)K_2^\prime(t_2)\text{ if }t_1<t_2\\
 G^{\prime\prime}\left (K_1(t_1)+H_2(t_2)\right )K_1^\prime (t_1)H_2^\prime(t_2)\text{ if }t_1>t_2\\
\end{array}\right .$$
 it is a straightforward computation to show that
$$\mathbb P(M_1\leq t,M_2\leq t,M_2>M_1)=G(\hat H(t))-G(K_2(t))-\int_0^{t}H_1^\prime(x)G^\prime \left (
\hat H(x)\right )dt$$
and 
$$\mathbb P(M_1\leq t,M_2\leq t,M_2<M_1)=1-G(K_1(t))+\int_0^{t}K_1^\prime(x)G^\prime \left (
\hat H(x)\right )dt$$
and
$$F_{M_1,M_2}^s(t)=\mathbb P(M_1\leq t,M_2\leq t,M_1=M_2)=-\int_{0}^{t}H^\prime_3(x)G^\prime \left (
\hat H(x)\right )dt.$$
\end{proof}
As a consequence
\begin{align*}\mathbb P(M_1=M_2)&=-\int_{0}^{+\infty}H^\prime_3(t)G^\prime \left (
\hat H(t)\right )dt\\
&=-\int_{0}^{\hat H^{-1}(x_G)}H^\prime_3(t)G^\prime \left (
\hat H(t)\right )dt\\
&=-\int_{0}^{\hat H^{-1}(x_G)}\frac{H^\prime_3(t)}{\hat H^\prime(t)}\hat H^\prime(t)G^\prime \left (
\hat H(t)\right )dt\\
&=\mathbb E\left [\frac{H_3^\prime(M)}{\hat H^\prime(M)}\right ].
\end{align*}
\begin{example}\label{ex1}
\begin{enumerate}
 \item If $G(x)=e^{-x}$, then
 $$F_{M_1,M_2}^s(t)=\int_{0}^{t}H^\prime_3(t)\exp\left (-\hat H(t)\right )dt
$$
and
$$\mathbb P(M_1=M_2)=\int_{0}^{+\infty}H^\prime_3(t)\exp\left (-\hat H(t)\right )dt
$$
and this is the case of the Generalized Marshall-Olkin distributions (see Li and Pellerey, 2011).
\item If $x_G=+\infty$ and $H_i=\lambda _iH$, then, if $\hat\lambda =\sum_{i=1}^3\lambda _i$,
 \begin{align*}
F_{M_1,M_2}^s(t)&=\frac{\lambda _3}{\hat\lambda}(1-G(\hat\lambda H(t)) \\
&=\frac{\lambda _3}{\hat\lambda}\mathbb P(M\leq t)
 \end{align*}
and 
$$\mathbb P(M_1=M_2)=\frac {\lambda _3}{\hat\lambda }.$$
independently of $G$. In particular, this case includes, when $G(x)=e^{-x}$, the Marshal-Olkin type distribution introduced by
Muliere and Scarsini (1987), and, if $H_i(x)=\lambda _ix$, the standard Marshall-Olkin case distribution.
\item More in general, let us assume that $H_3$ is proportional to $H_1+H_2$, that is $H_3=c(H_1+H_2)$, which is equivalnet to:
$\bar F_{\min (M_1,M_2)}(t)=G\left (\frac 1cH_3(t)\right )$. 
If $x_G<+\infty$, necessarily
$c=\frac 12$, while any choice of $c>0$ is allowed if $x_G=+\infty$.
In any case, we have
\begin{align*}
 F_{M_1,M_2}^s(t)&=\frac c{c+1}\left (1-G\left (\frac c{c+1} H_3(t)\right )\right )\\
&=\frac c{c+1}\mathbb P(M\leq t)
 \end{align*}
and
$$\mathbb P(M_1=M_2)=\frac {c}{c+1}.$$
that doesn't depend on $G$.
\item Both the facts that $F_{M_1,M_2}^s$ is proportional to the cumulative distribution function of $M$ and
that the singular component is independent of $G$ constitute a peculiarity of the above case.\\ If $G(x)=(x+1)^{-1/\theta}$ (Clayton generator) with $\theta >0$ and $H_1(x)+H_2(x)=H_3^2(x)+H_3(x)$
then
$$F_{M_1,M_2}^s(t)=\frac 1{2+\theta}\left (1-(1+H_3(t))^{-\frac 2\theta -1}\right )$$
but
$$\mathbb P(M\leq t)=1-(1+H_3(t))^{-\frac 2\theta}.$$
Moreover,this is equivalent to
$$\mathbb P(M_1=M_2)=\frac {1}{2+\theta}$$
that depends on $\theta$.\\
Similarly, if $H_1(x)+H_2(x)=e^{H_3(x)}-H_3(x)-1$ 
then, again,
$$F_{M_1,M_2}^s(t)=\frac 1{\theta +1}\left (1-e^{-\frac{\theta +1}\theta H_3(t)}\right )$$
with
$$\mathbb P(M\leq t)=1-e^{-\frac 1\theta H_3(t)}$$
and
$$\mathbb P(M_1=M_2)=\frac {1}{1+\theta}.$$
\end{enumerate}
\end{example}

\section{The Archimedean-based Marshall-Olkin Copula Function}\label{AMOCopula}

If $u_i=G(K_i(t_i))$, from (\ref{margin}), we get 
$t_i=K_i^{-1}(G^{-1}(u_i))$. 
Substituting it in (\ref{surv1}), thanks to Sklar's theorem, we obtain the copula associated to the vector $(M_1,M_2)$, that is
\begin{equation}\label{extended_copula}
\begin{aligned}
 &\bar C_{M_1,M_2}(u_1,u_2)=
&=\left\{\begin{array}{c}
G(H_1\circ K_1^{-1}(G^{-1}(u_1))+G^{-1}(u_2))\text{ if }K_1^{-1}(G^{-1}(u_1))\leq K_2^{-1}(G^{-1}(u_2))\\
G(G^{-1}(u_1)+H_2\circ K_2^{-1}(G^{-1}(u_2)))\text{ if }K_1^{-1}(G^{-1}(u_1))> K_2^{-1}(G^{-1}(u_2))\\
         \end{array}
\right .
\end{aligned}\end{equation}
Setting $D_i(x)=H_i\circ K_i^{-1}(x)$ for $x\in [0,x_G]$ and $i=1,2$ we have
$$\bar C_{M_1,M_2}(u_1,u_2)=
\left\{\begin{array}{c}
   G\left (D_1(G^{-1}(u_1))+G^{-1}(u_2)\right )
\text{ if }K_1^{-1}(G^{-1}(u_1))\leq K_2^{-1}(G^{-1}(u_2)))\\
G\left (G^{-1}(u_1)+D_2(G^{-1}(u_2))\right )\text{ if }
K_1^{-1}(G^{-1}(u_1))> K_2^{-1}(G^{-1}(u_2)))\\   
      \end{array}
\right .
$$
Let us analyze the set 
\begin{equation}\label{frontier}
\begin{aligned}F&=\{(u_1,u_2)\in [0,1]^2:K_1^{-1}(G^{-1}(u_1))=K_2^{-1}(G^{-1}(u_2))\}=\\
 &=
\{(u_1,u_2)\in [0,1]^2:
D_1(G^{-1}(u_1))+G^{-1}(u_2)=G^{-1}(u_1)+D_2(G^{-1}(u_2))
\}.
\end{aligned}\end{equation}
We have
\begin{itemize}
 \item if $D_1(x_G)>D_2(x_G)$, then $[0,K_1^{-1}(x_G)]\subsetneqq [0,K_2^{-1}(x_G)]$ and the set (\ref{frontier}) is the graph of the function
$u_2=h(u_1)=G\circ K_2\circ K_1^{-1}\circ G^{-1}(u_1)$ for $u_1\in [0,1]$, where
$h$ is strictly increasing, $h(1)=1$ and $h(0)=\mathbb P(M_2>K_1^{-1}(x_G))>0$;
 \item if $D_1(x_G)<D_2(x_G)$, then
$[0,K_1^{-1}(x_G)]\supsetneqq [0,K_2^{-1}(x_G)]$ and the set (\ref{frontier}) is the graph of the function
$u_1=\hat h(u_2)=G\circ K_1\circ K_2^{-1}\circ G^{-1}(u_2)$ for $u_2\in [0,1]$, where 
$\hat h$ is strictly increasing, $\hat h(1)=1$ and $\hat h(0)=\mathbb P(M_1>K_2^{-1}(x_G))>0$;
 \item if $D_1(x_G)=D_2(x_G)$, then the set (\ref{frontier}) can be represented through both $h$ and $\hat h$ that 
are strictly increasing, $h(1)=\hat h(1)=1$ and $h(0)=\hat h(0)=0$.
\end{itemize}
It follows that
\begin{itemize}
    \item if $D_1(x_G)\geq D_2(x_G)$, 
\begin{equation}\label{c1}\bar C_{M_1,M_2}(u_1,u_2)=
\left\{\begin{array}{c}
   G\left (D_1(G^{-1}(u_1))+G^{-1}(u_2)\right )
\text{ if }u_2\leq h(u_1)\\
G\left (G^{-1}(u_1)+D_2(G^{-1}(u_2))\right )\text{ if }
u_2> h(u_1)\\   
      \end{array}
\right .
\end{equation}
\item if $D_1(x_G)< D_2(x_G)$,
\begin{equation}\label{c2}\bar C_{M_1,M_2}(u_1,u_2)=
\left\{\begin{array}{c}
   G\left (D_1(G^{-1}(u_1))+G^{-1}(u_2)\right )
\text{ if }u_1\geq \hat h(u_2)\\
G\left (G^{-1}(u_1)+D_2(G^{-1}(u_2))\right )\text{ if }
u_1< \hat h(u_2)\\   
      \end{array}
\right .
\end{equation}

   \end{itemize}
Clearly, we recover a copula of type  (\ref{c2}) from a copula of type (\ref{c1}) by inverting $D_1$ with $D_2$ 
and viceversa.

\begin{remark}
Notice that $D_i(x)$ and $x-D_i(x)$ are continuous and strictly increasing on $[0,x_G]$.
Moreover
\begin{equation}\label{deform}
D_i(x)< x,\text{ for }x>0
\end{equation}
and, if $x_G=+\infty$, $\underset{x\rightarrow x_G}\lim D_i(x)=+\infty$ and $\underset{x\rightarrow x_G}\lim x-D_i(x)=+\infty$.
\end{remark}
In next result we will prove that for any suitable pair of functions $D_1$ and $D_2$, the copula in
(\ref{c1}) and (\ref{c2}) can be obtained starting from a three-variate distribution function of 
type (\ref{startdistr}) and following the construction presented above.                                                                                                                             
\smallskip

In what follows we set $\hat D_i(x)=x-D_i(x)$. The following result holds: 
\begin{lemma}\label{l1}
Let $D_1,D_2$ be two strictly increasing and continuous functions on $[0,x_G]$ with $D_i(0)=0$,
$\hat D_i(x)$ strictly increasing and such that, if $x_G=+\infty$, $\underset{x\rightarrow x_G}\lim D_i(x)=
D_i(x_G)=+\infty$ and $\underset{x\rightarrow x_G}\lim \hat D_i(x)=
\hat D_i(x_G)=+\infty$. Then, for every $H_3$ defined on $[0,+\infty)$, continuous,
strictly increasing and such that,
$H_3(0)=0$ and $H_3(+\infty)=x_G$, the function $H_i(x)=\hat D_i^{-1}(H_3(x))-H_3(x)$, for $x\in[0,H_3^{-1}(\hat D_i(x_G)]$, is the unique solutions of
\begin{equation}\label{relazio}D_i(x)=H_i\circ (H_i+H_3)^{-1}(x)\text{ for all }x\in[0,x_G].\end{equation}
\end{lemma}
\begin{proof}
Notice that $H_i(x)$ is strictly increasing on $[0,H_3^{-1}(\hat D_i(x_G))]$ and satisfies $H_i(0)=0$, 
$H_i(H_3^{-1}(\hat D_i(x_G)))=D_i(x_G)$ and it is the only solution of 
$$\hat D_i(H_i(z)+H_3(z))=H_3(z)\qquad z\in [0,H_3^{-1}(\hat D_i(x_G))]$$
which is equivalent to 
\begin{equation}\label{cond1}D_i(H_i(z)+H_3(z))=H_i(z)\qquad z\in [0,H_3^{-1}(\hat D_i(x_G))]\end{equation}
Setting $x=H_i(z)+H_3(z)\in [0,x_G]$ in (\ref{cond1}) we get (\ref{relazio}). 
\\
If $x_G<+\infty$, $H_i$ can be clearly extended so to be defined on the set $[0,+\infty)$, to be strictly increasing and
with $H_i(+\infty )=x_G$.

\end{proof}
\begin{remark}
Notice that $D_i$ and $\hat D_i$ satisfy the same assumptions and if, given $H_3$, $H_i$ solves (\ref{relazio}), then
$\hat H_i(x)=D_i^{-1}(H_3(x))-H_3(x)=\left (\hat D_i^{-1}-Id\right )^{-1}(H_3(x))$ (where $Id$ is the identity function) solves
$\hat D_i(x)=H_i\circ (H_i+H_3)^{-1}(x)$ for all $x\in[0,x_G]$.
\end{remark}

\begin{example}\label{ex2}\begin{enumerate}
\item If $D_i(x)=\alpha _ix$, $\alpha _i\in (0,1)$, $i=1,2$, then 
$H_i(x)=\frac{\alpha _i}{1-\alpha _i}H_3(x)$, for $i=1,2$ and $x\in [0,H_3^{-1}((1-\alpha _i)x_G)]$.
\item If $D_i(x)=x-\frac 1{\alpha _i}\log (\alpha_i x+1)$, $\alpha _i>0$, $i=1,2$, then 
$H_i(x)=\frac 1{\alpha _i}\left (e^{\alpha _iH_3(x)}-1\right )-H_3(x)$, for $i=1,2$, and $x\in [0,H_3^{-1}(\frac 1{\alpha _i}\log (\alpha_ix_G+1))]$.
\item If $D_i(x)=x-\frac 1{\alpha _i}\left (\sqrt {\alpha_i x+1}-1\right )$, $\alpha _i>0$, $i=1,2$, then 
$H_i(x)=\alpha _iH_3^2(x)+H_3(x)$, for $i=1,2$ and $x\in [0,H_3^{-1}(\frac 1{\alpha _i}(\sqrt{\alpha _i x_G+1}-1))]$.
\item If $D_i(x)=\frac 1{\alpha _i}\left (\sqrt {\alpha_i x+1}-1\right )$, $\alpha _i>0$, $i=1,2$, then 
$H_i(x)=\frac 1{2\alpha _i}\left (\sqrt{1+4\alpha _iH_3(x)}-1\right )$, for $i=1,2$ and $x\in [0,H_3^
{-1}(x_G-\frac 1{\alpha _i}(\sqrt{\alpha _i x_G+1}-1))]$.
\end{enumerate}
\end{example}

\begin{definition}\label{DefCAMO}
Let $G:[0,+\infty)\rightarrow [0,1]$, with $G(0)=1$ and such that 
$G^\prime$ exists on $(0,+\infty)$,
 it is non-positive, non-decreasing and concave. If $x_G=\inf\{x\geq 0:G(x)=0\}$, 
let $D_1,D_2$ be two strictly increasing and continuous functions on $[0,x_G]$ 
with $D_i(0)=0$, 
$\hat D_i(x)$ strictly increasing and such that if $x_G=+\infty$, 
$\underset{x\rightarrow x_G}\lim D_i(x)=D_i(x_G)=+\infty$ and $\underset{x\rightarrow x_G}\lim \hat D_i(x)=
\hat D_i(x_G)=+\infty$. If we set 
$f(x)=\hat D_2^{-1}\circ \hat D_1(x)$ (when this is well defined), we call \emph{Archimedean-based Marshall-Olkin copula with generator $G$
and distortions $D_1$, $D_2$} the copula function $C_{AMO}$ so defined:
\begin{itemize}
    \item if $D_1(x_G)\geq D_2(x_G)$, 
\begin{equation}\label{case1}C_{AMO}(u,v)=
\left\{\begin{array}{c}
   G\left (D_1(G^{-1}(u))+G^{-1}(v)\right )
\text{ if }v\leq h(u)\\
G\left (G^{-1}(u)+D_2(G^{-1}(v))\right )\text{ if }
v> h(u)\\   
      \end{array}
\right .
\end{equation}
where $h(u)=G\circ f\circ G^{-1}(u)$;
\item if $D_1(x_G)< D_2(x_G)$,
\begin{equation}\label{case2}C_{AMO}(u,v)=
\left\{\begin{array}{c}
   G\left (D_1(G^{-1}(u))+G^{-1}(v)\right )
\text{ if }u\geq \hat h(v)\\
G\left (G^{-1}(u)+D_2(G^{-1}(v))\right )\text{ if }
u_1< \hat h(u_2)\\   
      \end{array}
\right .
\end{equation}
where $\hat h(v)=G\circ f^{-1}\circ G^{-1}(v)$.
   \end{itemize}
\end{definition}\bigskip

It is evident from the definition that the $C_{AMO}$ copula is a distortion of the Archimedean copula 
with generator $G$ through the functions $D_i$.
More precisely, if $D_1\neq D_2$ (that is if $H_1\neq H_2$), the Archimedean copula is differently modified above and below the curve $F$ given in (\ref{frontier})
and the obtained copula is obviously asymmetric.\\
Conversely, if $D_1=D_2=D$ (that is if $H_1=H_2$) we get
\begin{equation}\label{durante}
C_{AMO}(u,v)
=\left\{\begin{array}{c}
   G\left (D(G^{-1}(u))+G^{-1}(v)\right )\text{ if }v\leq u\\
G\left (G^{-1}(u)+D(G^{-1}(v))\right )\text{ if }v>u\\   
      \end{array}
\right .
\end{equation}
and the obtained copula remains exchangeable. The family defined in (\ref{durante}) is contained in
the class of copulas introduced in Durante et al. (2007), defined
as
\begin{equation}\label{d} C_{\phi,\psi}(u,v)=\phi^{[-1]}(\phi(u\wedge v)+\psi (u\vee v))\end{equation}
with $\phi:[0, 1]\rightarrow [0,+\infty]$, continuous, convex
and strictly decreasing, $\psi:[0, 1]\rightarrow [0,+\infty]$, continuous, decreasing and such that
$\psi(1)=0$ and $\psi-\phi$ increasing in $[0,1]$:
(\ref{durante}) is trivially recovered from (\ref{d}) when $\phi (1)=G^{-1}(1)=0$ and $\psi(t)=D(G^{-1}(t))$.
\begin{example}
 If $G(x)=e^{-x}$
$$C_{AMO}(u,v)=\left \{\begin{array}{c}
ve^{-D_1(-\ln u)}\text{ if }v\leq \exp (-f(-\ln u))\\ 
ue^{-D_2(-\ln v)}\text{ if }v> \exp (-f(-\ln u))\\
             \end{array}
\right .
$$
and we recover the Generalized Marshall-Olkin copula introduced in Li and Pellerey (2011).
In particular, if, for $i=1,2$, $D_i(x)=\alpha_ix$ with $\alpha _i\in (0,1)$, we recover the classical Marshal-Olkin copula.\end{example}
\begin{example}
If $D_i(x)=\alpha_ix$ for $i=1,2$ and $\alpha _i\in (0,1)$, then
\begin{equation}\label{constant}C_{AMO}(u,v)
  =\left\{\begin{array}{c}
          G\left (\alpha _1G^{-1}(u)+G^{-1}(v)\right )\text{ if }\,v\leq G\left (\frac{1-\alpha _1}{1-\alpha _2}G^{-1}(u)\right )\\
          G\left (G^{-1}(u)+\alpha _2G^{-1}(v)\right )\text{ if }\,
v> G\left (\frac{1-\alpha _1}{1-\alpha _2}G^{-1}(u)\right )         \end{array}
\right .\end{equation}
This particular specification constitutes a subclass of the family
 of copulas called  Archimax copulas introduced in Cap\'{e}ra\`{a} et al. (2000). Archimax copulas are defined as
\begin{equation}\label{archimax}C_{G,A}(u,v)=G\left (\left (G^{-1}(u)+G^{-1}(v)\right ) A\left (\frac{G^{-1}(u)}{G^{-1}(u)+G^{-1}(v)}\right )\right )\end{equation}
where $A:[0,1]\rightarrow [1/2,1]$ is a convex function such that $\max(t, 1-t)\leq A(t)\leq 1$ for all
$t\in [0,1]$ and copulas in (\ref{constant}) can be obtained choosing
$$A(t)=\left \{\begin{array}{c}
             (\alpha _1-1)t+1\text{ if }t\leq \frac{1-\alpha _2}{2-\alpha_1-\alpha _2}\\
             (1-\alpha _2)t+\alpha _2\text{ if }t> \frac{1-\alpha _2}{2-\alpha_1-\alpha _2}\\
            \end{array}
\right .$$ 
The same fact holds if one considers piecewise linear distortions $D_i$, to which correspond piecewise affine 
functions $A$ in representation (\ref{archimax}).
\medskip

Moreover, the subset of copulas of type  (\ref{constant}) when $G$ is the Laplace transform of a positive random variable, coincides with
the family of bivariate Scale Mixture of Marshall-Olkin copulas (see Li (2009), Bernhart et al. (2013) and
Mai et al. (2013)). These copulas are obtained as the copulas associated to the random vector $\left (\frac{Z_1}Y,\frac{Z_2}Y\right)$ where
$(Z_1,Z_2)$ has a bivariate Marshall-Olkin distribution and $Y$ is an independent positive random variable. 
In the exchangeable case, May et al. (2013) propose an alternative construction: 
let
$\epsilon _1,\epsilon _2$ be i.i.d. unit exponentially distributed random variables, $M>0$ be a random variable with
Laplace transform $G(x)$ and $\Lambda _t\neq 0$ be a L\'evy subordinator with Laplace exponent $\Psi$; 
assuming that they are mutually independent, define
$$\tau_k=\inf\{t\geq 0:M_t\geq \epsilon _k\}$$
where $M_t=\Lambda_{MG^{-1}(1-p(t))/\Psi (1)}$ with $p(t)$ a given distribution function. 
The copula associated to $(\tau_1,\tau _2)$
is
$$C(u,v)=G \left(G^{-1}(\min (u,v))+ G^{-1}(\max (u,v))\frac{\Psi (2)-\Psi (1)}{\Psi (1)}\right )$$
which is obviously of type (\ref{constant}) with $\alpha _1=\alpha _2=\frac{\Psi (2)-\Psi (1)}{\Psi (1)}$.

\end{example}
\bigskip

Under the assumptions of Proposition \ref{singularity},
$$\begin{aligned}
   -\int_{0}^{\hat H^{-1}(x_G)}H^\prime_3(t)G^\prime \left (
\hat H(t)\right )dt&=\mathbb P(M_1=M_2)=\\
&=\mathbb P(\bar F_{M_1}^{-1}(U)=\bar F_{M_2}^{-1}(V))=\\
&=\mathbb P(K_1^{-1}\circ G(U)=K_2^{-1}\circ G(V))=\\
&=\mathbb P(F).
  \end{aligned}
$$
Hence, if $G$ is twice differentiable and each $D_i$ is differentiable on $[0,x_G]$,
the copula function $C_{AMO}(u,v)$ has
a singularity on the curve $F$ (see (\ref{frontier})), whose mass, expressed in terms of the distortions, is 
$$\mathbb P(F)=-\int_{0}^{T^{-1}(x_G)}G^\prime(T(x))
dx$$
where 
\begin{equation}\label{ti}
 T(x)=\hat D_1^{-1}(x)+\hat D_2^{-1}(x)-x
\end{equation}
\bigskip

It is a known fact that two Archimedean copulas with generators $G_A$ and $G_B$ coincide if and only if
there exists $\alpha >0$ such that $G_A(x)=G_B(\alpha x)$ (see Corollary 2.2.6 in Alsina et al., 2006). In what follows we will present the analogous result for the Archimedean-based 
Marshall-Olkin copula functions.
\begin{theorem}\label{coincident}
Let $C^A$ and $C^B$ be two Archimedean-based Marshall-Olkin copula functions with generators $G_A$ and $G_B$ and distortions
$A_1,\,A_2$ and $B_1,\, B_2$, respectively.
$$C^A\equiv C^B$$ if and only if there exists $m>0$ such that
\begin{equation}\label{sab22}G_B(z)=G_A(mz)\text{ for  }z\in [0,+\infty)\end{equation}
and
\begin{equation}\label{sab33}
B_i(z)=\frac 1mA_i(mz)\text{ for }z\in \left[0,x_{G_B}\right ],\, i=1,2.\end{equation}
\end{theorem}
\begin{proof}
Since we are looking for $C^A\equiv C^B$, necessarily the corresponding singularity sets $F$ in (\ref{singularity})
must coincide. Hence $C^A$ and $C^B$ must be of the same type (\ref{case1}) or (\ref{case2}).
We show the proof in case (\ref{case1}) being the other perfectly identical.
Hence, we assume $A_1(x_{G_A})\geq A_2(x_{G_A})$, $B_1(x_{G_B})\geq B_2(x_{G_B})$ and 
$$C^A(u,v)
=\left\{\begin{array}{c}
   G_A\left (A_1(G_A^{-1}(u))+G_A^{-1}(v)\right )\text{ if }v\leq h_A(u)\\
   G_A\left (G_A^{-1}(u)+A_2(G_A^{-1}(v))\right )\text{ if }v> h_A(u)\\
      \end{array}
\right .
$$
and
$$C^B(u,v)
=\left\{\begin{array}{c}
   G_B\left (B_1(G_B^{-1}(u))+G_B^{-1}(v)\right )\text{ if }v\leq h_B(u)\\
   G_B\left (G_B^{-1}(u)+B_2(G_B^{-1}(v))\right )\text{ if }v> h_B(u)\\
      \end{array}
\right ..
$$
Let $h(u)=\min(h_A(u),h_B(u))$. For $C^A\equiv C^B$, necessarily 
\begin{equation}\label{sab1}G_A\left (A_1(G_A^{-1}(u))+G_A^{-1}(v)\right )=G_B\left (B_1(G_B^{-1}(u))+G_B^{-1}(v)\right )\end{equation}
for $u\in [0,1]$ and $v\leq h(u)$. We set $\psi=G_A^{-1}\circ G_B:[0,x_{G_B}]\rightarrow [0,x_{G_A}]$, 
$x=G_B^{-1}(u)\in [0,x_{G_B}]$ and $y=G_B^{-1}(v)\in [0,x_{G_B}]$. So (\ref{sab1}) can be rewritten as
\begin{equation}\label{sab2}
 A_1(\psi(x))+\psi(y)=\psi(B_1(x)+y)
\end{equation} 
with $y\in[0,x_{G_B}]$ and $x\leq G_B^{-1}(h^{-1}(G_B(y)))$. Since $G_A$ and $G_B$ are differentiable on $(0,x_{G_A})$
and $(0,x_{G_B})$, respectively, and their derivatives
are strictly negative, $\psi$ is differentiable on $(0,x_{G_B})$ as well. After differentiating (\ref{sab2}) 
with respect to $y$, we get
$$\psi^\prime (y)=\psi^\prime(B_1(x)+y).$$
This holds whenever  $0<x<G_B^{-1}(h^{-1}(G_B(y)))$; hence $\psi^\prime$, being locally constant,
is also globally contant:
then there exists $m$ such that
$\psi^\prime(y)=m$ for all $y\in (0,x_{G_B})$. It follows that $\psi(y)=my+q$, for $y\in [0,x_{G_B}]$, and since 
$\psi$ is strictly increasing, necessarily,
$m> 0$. Moreover, since $\psi(0)=0$, $q=0$.\\
Hence
\begin{equation}\label{sab11}G_B(x)=G_A(mx),\, x\in [0,x_{G_B}]\end{equation}
and $x_{G_B}=\frac{x_{G_A}}m$.\\
Substituting this relation in (\ref{sab1}) we get
$$A_1(G_A^{-1}(u))=mB_1\left (\frac{G_A^{-1}(u)}m\right)$$
and, setting $z=\frac{G_A^{-1}(u)}m\in\left[0,\frac{x_{G_A}}m\right]$, we obtain
$$B_1(z)=\frac 1mA_1(mz).$$
In order to guarantee that $C^A\equiv C^B$ 
$$G_A\left (G_A^{-1}(u)+A_2(G_A^{-1}(v))\right )=G_B\left (G_B^{-1}(u)+B_2(G_B^{-1}(v))\right )$$
must hold for $v\in [0,1]$ and $u< \min(h^{-1}_A(v),h^{-1}_B(v))$, as well.
Using (\ref{sab11}) we obtain
$$A_2(G_A^{-1}(v))=mB_2\left (\frac{G_A^{-1}(v)}m\right)$$
and, setting $z=\frac{G_A^{-1}(v)}m\in\left[0,\frac{x_{G_A}}m\right ]$, we obtain
$$B_2(z)=\frac 1mA_2(mz).$$
\medskip

The converse is trivial.
\end{proof}

\section{Dependence Properties}\label{dp}
In this Section we will consider the concordance measure induced by the $C_{AMO}$ copula by calculating the corresponding
Kendall's function and Kendall's tau. Since, as we noticed,
in general, the $C_{AMO}$ copula is an asymmetric distortion of an Archimedean copula, we investigate the effect of
such a distortion on the tail dependence parameters.\medskip

Let us start by observing that, thanks to (\ref{deform}),
\begin{equation}\label{maggiorazione}
 C_{AMO}(u,v)\geq G(G^{-1}(u)+G^{-1}(v)), \qquad (u,v)\in [0,1]^2.
\end{equation}

\subsection{The Kendall's function}
We remind that the Kendall's function of a copula $C$, is the cumulative distribution function of the random variable $C(U,V)$
with respect to the probability induced by $C$ (see Nelsen, 2006), that is the $C$-measure of the set
$$A_t=\{(u,v)\in [0,1]^2:C_{AMO}(u,v)\leq t\}.$$
Let $K_G(t)$ be the Kendall's function of the Archimedean copula with generator $G$ and
$K_{AMO}$ be the Kendall's function of the copula $C_{AMO}$. Clearly, from (\ref{maggiorazione})
$$K_{AMO}(t)\leq K_G(t),\qquad t\in [0,1].$$
More precisely,
\begin{theorem}
$$K_{AMO}(t)=K_G(t)+
G^\prime( G^{-1}(t))\cdot
T^{-1} (G^{-1}(t))$$
where $T$ is defined in (\ref{ti}).  
\end{theorem}

\begin{proof}
We follow the same ideas and spirit of the proof of the analogous result in the Archimedean case (see, for example, 
Theorem 4.3.4 in Nelsen, 2006).\medskip

We will prove the result in the case $D_1(x_G)\geq D_2(x_G)$, being the alternative case perfectly analogous.\\
We start with looking for the intersection of the $t$-level curve $\{(u,v)\in [0,1]^2:C_{AMO}(u,v)=t\}$, for $t>0$,
and the graph of the function $h$ of Definition \ref{DefCAMO}, that is we solve for $u$ the equation
$$G(D_1( G^{-1}(u))+G^{-1}(h(u)))
=t$$
from which
$$D_1(G^{-1}(u))+f(G^{-1}(u))=G^{-1}(t)$$
or, setting $g(x)=D_1(x)+f(x)$,
$$g(G^{-1}(u))=G^{-1}(t).$$
Since, thanks to the assumptions, $g$ is invertible, the unique solution is 
\begin{equation}\label{ut}u_t=G\circ g^{-1}\left (G^{-1}(t)\right )\in [t,1]\end{equation}
and, if
\begin{equation}\label{vt}v_t=h(u_t)=G\circ f\circ g^{-1}\left (G^{-1}(t)\right ),\end{equation}
$(u_t,v_t)$ is the unique intersection of the level curve and the graph of $h$.
\bigskip

We split the set $A_t$ into three regions: the rectangle $R_t=[0,u_t]\times [0,v_t]$ and the sets 
$$B_1=\{(u,v)\in [0,1]^2:u\in [u_t,1], v\leq G(G^{-1}(t)-D_1(G^{-1}(u)))\}$$ and 
$$B_2=\{(u,v)\in [0,1]^2:v\in [v_t,1], u\leq G(G^{-1}(t)-D_2(G^{-1}(v)))\}$$
 
Obviously 
$$\mathbb P(R_t)=C(u_t,v_t)=t.$$
Let us compute now $\mathbb P(B_1)$ and $\mathbb P(B_2)$. We start with $\mathbb P(B_1)$.
\medskip

Consider the partition of the interval $[u_t,1]$ given by the points 
$$t_k=G\circ D_1^{-1}\left (D_1\left (G^{-1}(u_t)\right ) (1-\frac k{n} )\right ).$$
Notice that $t_0=u_t\leq t_1\leq \cdots\leq t_{n}=1$ and that, since $G\circ D_1^{-1}$ is uniformly
continuous on $[0,D_1(G^{-1}(u_t))]$, the width of the partition so defined goes to zero as $n\rightarrow +\infty$.
\medskip

For $u\in[u_t,1]$, the $t$-level curve in $B_1$ is
$$v(u)=G\left (G^{-1}(t)-D_1(G^{-1}(u)) \right ).$$

For $k=1,\ldots ,n$, let $R_k=[t_{k-1},t_{k}]\times [0,v(t_{k-1})]$. Clearly
$$P(R_k)=C(t_k,v(t_{k-1}))-C(t_{k-1},v(t_{k-1}))=
C(t_k,v(t_{k-1}))-t$$
and
$$\begin{aligned}
   C(t_k,v(t_{k-1}))&=
G\left (D_1(G^{-1}(t_k))+G^{-1}(v(t_{k-1}))\right )=\\
&=G\left (G^{-1}(t)-\frac{D_1(G^{-1}(u_t))}{n}\right ).
  \end{aligned}
$$
If $S_n=\sum_{k=1}^{n}R_k$,
$$\begin{aligned}\mathbb P(S_n)&=n\left [G\left (G^{-1}(t)-\frac{D_1(G^{-1}(u_t))}{n}\right )-t\right ]=\\
  &=\frac{\left [G\left (G^{-1}(t)-\frac{D_1(G^{-1}(u_t))}{n}\right )-G(G^{-1}(t))\right ]}
{\frac{D_1(G^{-1}(u_t))}{n}}D_1(G^{-1}(u_t))
  \end{aligned}
$$
and 
$$\mathbb P(B_1)=\underset{n\rightarrow +\infty}\lim \mathbb P(S_n)=-G^\prime (G^{-1}(t))D_1(G^{-1}(u_t)).$$
\medskip

For the set $B_2$, considering the partition of the interval $[v_t,1]$ given by the points 
$$s_k=G\left (D_2^{-1}\left (D_2\left (G^{-1}(v_t)\right ) (1-\frac k{n} )\right )\right ),$$
exactly as done for $B_1$, we get
$$\mathbb P(B_2)=-G^\prime (G^{-1}(t))D_2(G^{-1}(v_t)).$$
\bigskip

It follows that, using (\ref{ut}) and (\ref{vt}),
$$\begin{aligned}\mathbb P(A_t)&=\mathbb P(R_t)+\mathbb P(B_1)+\mathbb P(B_2)=\\
   &=t-
G^\prime (G^{-1}(t))\left [D_1(G^{-1}(u_t))+D_2(G^{-1}(v_t))\right ]=\\
&=t-
G^\prime (G^{-1}(t))\left [
D_1( g^{-1}(G^{-1}(t)))+D_2\circ f(g^{-1}(G^{-1}(t) ))\right ]=\\
&=K_G(t)+G^\prime (G^{-1}(t))\left [G^{-1}(t)
-D_1( g^{-1}(G^{-1}(t)))-D_2\circ f(g^{-1}(G^{-1}(t) ))\right ]\\
&=K_G(t)+G^\prime (G^{-1}(t))\left [Id
-D_1\circ g^{-1}-D_2\circ f\circ g^{-1}\right ]\circ G^{-1}(t) \\
  \end{aligned}.$$
where $Id$ is the identity function and $K_G(t)=t-
G^\prime (G^{-1}(t))G^{-1}(t)$ is the Kendall's function of the Archimedean copula with generator $G$.\\
But, since $g(x)=T\circ \hat D_1(x)$, where $T$ is defined in (\ref{ti}), we have
$$\begin{aligned}
  Id-D_1\circ g^{-1}-D_2\circ f\circ g^{-1} &=(g-D_1-D_2\circ f)\circ g^{-1}=\\
&=(f-D_2\circ f)\circ g^{-1}=\\
&=\hat D_2\circ f\circ g^{-1}=\\
&=\hat D_1\circ g^{-1}=\\
&=T^{-1}\\
  \end{aligned} 
$$
and so
$$K_{AMO}(t)=K_G(t)+G^\prime (G^{-1}(t))T^{-1} (G^{-1}(t)).$$
\end{proof}

In terms of the generating functions $H_1,H_2,H_3$ the Kendall's function can be rewritten as
$$K_{AMO}(t)=K_G(t)+G^\prime (G^{-1}(t))\cdot H_3\circ \hat H^{-1}(G^{-1}(t)).$$
\bigskip

In particular, if $H_3=c(H_1+H_2)$, with $c>0$ (see Example \ref{ex1}) we get 
\begin{equation}\label{multiplo}K_{AMO}(t)=K_G(t)+\frac c{c+1}G^\prime (G^{-1}(t))G^{-1}(t)\end{equation}
and this case includes the case of linear distortions of Example \ref{ex2}.
\medskip

Notice that (\ref{multiplo}) can be rewritten as $K_{AMO}(t)=t-\frac 1{c+1}G^\prime (G^{-1}(t))G^{-1}(t)$, and if we consider two different generators $G_A$ and $G_B$ so that 
$K_{G_A}(t)\leq K_{G_B}(t)$, but identical proportional parameter $c$, then 
$K_{G_AMO}(t)\leq K_{G_BMO}(t)$, where with $K_{G_jMO}$ we denote the Kendall's function of the Archimedean-based
Marshall-Olkin copula with generator $G_j$, for $j=A,B$.\\
This fact doesn't continue to hold in general as shown in next example.
\begin{example}\label{inversion}
Let $D_1(x)=D_2(x)=1+x-\sqrt{1+2x}$. It follows that
$$K_{AMO}(t)=t-\frac 14
G^\prime (G^{-1}(t))
\left (\sqrt{1+4\cdot G^{-1}(t)}-1\right )^2
$$
If we consider the Frank generator $G_F(x)=-\frac 1\theta\ln\left (1+e^{-x}(e^{-\theta}-1)\right )$ 
with parameter $\theta =4$ and the Gumbel generator $G_G(x)=\exp\left (-z^{1/\gamma}\right )$ with parameter
$\gamma =2$we have
$$K_{G_F}(0.3)=0.497>0.480=K_{G_G}(0.3)$$
while
$$K_{G_FMO}(0.3)=0.341<0.380=K_{G_GMO}(0.3).$$
\end{example}
\begin{example} \label{caso}Let us consider 
 $D_i(x)=x-\frac 1{\alpha _i}(\sqrt{\alpha _i x+1}-1)$ (see Example \ref{ex2}). Then
$$K_{AMO}(t)=K_G(t)+
\frac{3G^\prime( G^{-1}(t))}{2(\alpha_1+\alpha _2)}
\cdot
\left (\sqrt{1+\frac 49(\alpha _1+\alpha _2)G^{-1}(t)}-1\right )$$

\end{example}

\subsection{Kendall's tau}
It is known that the Kendall's tau of a bivariate copula $C$ is given by
$$\tau_C=4E[C(U,V)]-1=3-4\int _0^1K(t)dt.$$
Hence, we get
$$\begin{aligned}\tau_{AMO}&=3-4\int _0^1K_{AMO}(t)dt=\\
&=3-4\int _0^1(K_G(t)+
G^\prime (G^{-1}(t))\cdot
T^{-1} (G^{-1}(t)))dt=\\
&=\tau _G-4\int_0^1G^\prime (G^{-1}(t))\cdot
T^{-1}(G^{-1}(t))dt=\\
&=\tau _G+4\int_0^{x_G}\left (G^\prime (x)\right )^2\cdot
T^{-1}(x)dx
    \end{aligned}
$$
where $\tau_{G}$ is the Kendall's tau of the Archimedean copula with generator $G$.
Obviously, as expected from (\ref{maggiorazione}),
$$\tau_{AMO}\geq \tau_{G}.$$
In terms of the generating functions $H_1,H_2,H_3$ the Kendall's tau can be rewritten as
$$\tau_{AMO}=\tau _G+4\int_0^{x_G}\left (G^\prime (x)\right )^2\cdot
H_3\circ \hat H^{-1}(x)dx.$$
\begin{example}
If $G(x)=e^{-x}$ we recover the case studied in Li and Pellerey(2011) and we get
$$\tau_{AMO}=4\int_0^{+\infty}e^{-2x}T^{-1}(x)dx.$$
\end{example}

\begin{example}
Let us consider the Clayton case, that is $G(x)=(x+1)^{-1/\theta}$, with $\theta >0$.\begin{enumerate}  
\item If $H_3=c(H_1+H_2)$, $c>0$, then
$$\tau_{AMO}=\tau _G+\frac c{c+1}\frac 2{2+\theta}=\tau_G+\tau_{IMO}\frac 2{\theta+2}$$
where $\tau_{IMO}=\frac c{c+1}$ is the Kendall's tau when $G=e^{-x}$. Notice that this result continues to hold for $\theta\in(-\frac 12,0)$ and $c=\frac 12$.
\item If $\hat D_1^{-1}(x)+\hat D_2^{-1}(x)=e^{\frac x\gamma}-1+x$, with $1\geq\gamma >0$  
(this corresponds to $H_1(z)+H_2(z)=e^{\frac{H_3(z)}\gamma}-H_3(z)-1$), we get
$\left (\hat D_1^{-1}(x)+\hat D_2^{-1}(x)-x\right )^{-1}=\gamma\ln (x+1)$. Then
$$\tau_{AMO}=\tau_G+\frac {4\gamma}{(2+\theta )^2}=\frac {\theta ^2+2\theta +4\gamma}{(2+\theta )^2}.$$
\item  
If $\hat D_1^{-1}(x)+\hat D_2^{-1}(x)=\left (\frac x\gamma +1\right )^{\alpha}-1+x$ with $\alpha >1$, $1\geq\gamma >0$ (this corresponds to 
$H_1(z)+H_2(z)=\left (\frac{H_3(z)}\gamma+1\right )^{\alpha}-1-H_3(z)$), we get
$\left (\hat D_1^{-1}(x)+\hat D_2^{-1}(x)-x\right )^{-1}=\gamma\left ((x+1)^{\frac 1{\alpha}}-1\right )$. 
Then
$$\tau_{AMO}=\tau_G+\frac {4\gamma}{(2+\theta )(\alpha(2+\theta)-\theta)}=\frac 1{\theta +2}\left (\theta+
\frac {4\gamma}{\alpha(2+\theta)-\theta}\right ).$$
\item If $D_i(x)=x-\frac 1{\alpha _i}(\sqrt{\alpha _i x+1}-1)$ (see Example \ref{caso}),
$\left (\hat D_1^{-1}(x)+\hat D_2^{-1}(x)-x\right )^{-1}=
\frac{3}{2(\alpha_1+\alpha _2)}
\left ((\sqrt{1+\frac 49(\alpha _1+\alpha _2)x}-1\right )$. If $\alpha_1+\alpha _2=\frac 94$, we recover the previous case with $\gamma=\frac 23$ and $\alpha =2$.
Hence
$$\tau_{AMO}=\frac {3\theta ^2+12\theta+8}{3(2+\theta)(4+\theta)}.$$

\end{enumerate}

\end{example}
Exactly as for the Kendall's function, the presence of distortion functions strongly influences the Kendall's tau.
In Table \ref{tab1}, the values of the Kendall's tau of some Archimedean copulas and the corresponding values
of the Archimedean-based Mashall-Olkin copulas with same generator but distortions $
D_1(x)=D_2(x)=1+x-\sqrt{1+2x}$ are reported in order to illustrate how the presence of the distortions  
can induce an inversion in the order of the concordance measure.

\begin{center}
 \begin{table}
\caption{\footnotesize Comparison of Kendall's tau value with and without distortions.}\label{tab1}\vskip 1pc
\begin{center}\begin{tabular}{|c|c|c|c|}
\hline
Generator&Parameter&$\tau_G $&$\tau_{AMO}$\\
\hline 
Clayton&$\theta=2$&$0.5$&$0.78539$\\
\hline
Gumbel&$\theta=1.8$&$0.375$&$0.81636$\\
\hline
Frank&$\theta=4$&$0.388$&$0.906$\\
\hline
\end{tabular}\end{center}
\end{table}
\end{center}

\subsection{Tail dependence}
We recall that the upper and lower tail dependence parameters of a copula $C$ are given, respectively, by
$$\lambda_U=2-\underset{u\uparrow 1}\lim\frac{1-C(u,u)}{1-u}\text{ and }
\lambda_L=\underset{u\downarrow 0}\lim\frac{C(u,u)}{u}$$
when these limits exist and that if $\lambda_U\in (0,1]$ ($\lambda_L\in (0,1]$) we have that $C$ has upper(lower)-tail dependence
(we refer to Section 5.4 in Nelsen (2006), for more details).\medskip

Let $\lambda ^{AMO}_L$ and $\lambda ^{AMO}_U$ be the lower and upper tail dependence parameters of the copula $C_{AMO}$ and 
$\lambda ^G_L$ and $\lambda ^G_U$ be the tail dependence parameters of the corresponding bivariate 
Archimedean copula with generator $G$.

Let $\mathcal A_1=\{x\in[0,x_G):D_1(x)\geq D_2(x)\}$ and $\mathcal A_2=\mathcal A_1^c$. Notice that 
$G^{-1}(u)\in\mathcal A_1$ if and only if $(u,u)$ lies on or below the curve $F$ defined in (\ref{frontier}).

\subsubsection{$\lambda^{AMO}_L$}
If $x_G<+\infty$, obviously $\lambda^{AMO}_L=0.$ 
\medskip

\noindent Let us now consider the case $x_G=+\infty$. 
We have 
$$C_{AMO}(u,u)=G\left (G^{-1}(u)+D_1(G^{-1}(u)){\bf 1}_{\mathcal A_1}(G^{-1}(u))+D_2(G^{-1}(u)){\bf 1}_
{\mathcal A_2}(G^{-1}(u))\right ).$$
Setting $x=G^{-1}(u)$, 
$$\begin{aligned}
 \lambda^{AMO}_L&= \underset{u\downarrow 0}\lim\frac{C_{AMO}(u,u)}u=\\
&=\underset{x\rightarrow +\infty}\lim\frac{G(x+D_1(x){\bf 1}_{\mathcal A_1}(x)+
D_2(x){\bf 1}_{\mathcal A_2}(x))}{G(x)}
  \end{aligned}
$$

\noindent
If there exists $\bar x>0$ such that $(\bar x,+\infty)\subset \mathcal A_i$ then 
$$ \lambda^{AMO}_L=\underset{x\rightarrow +\infty}\lim\frac{G(x+D_i(x))}{G(x)}.$$
The value of the above limit clearly depends on the type of decay of $G$ to zero. If the decay is of polinomial type, that is
there exist $c,\gamma >0$ such that $G(x)\underset{x\rightarrow +\infty}\sim cx^{-\gamma}$, then
$$\underset{x\rightarrow +\infty}\lim\frac{G(x+D_i(x))}{G(x)}=\underset{x\rightarrow +\infty}\lim\left (1+\frac{D_i(x)}x\right )^{-\gamma}$$
Hence, if $\underset{x\rightarrow +\infty}\lim \frac{D_i(x)}{x}=\beta\in [0,1]$ exists, we have that
$$ \lambda^{AMO}_L=(1+\beta)^{-\gamma}.$$
If the decay is exponential of type $G(x)\underset{x\rightarrow +\infty}\sim ce^{-ax^\gamma}$, with $a,c,\gamma>0$,
we have
$$\begin{aligned}\underset{x\rightarrow +\infty}\lim\frac{G(x+D_i(x))}{G(x)}&=
\underset{x\rightarrow +\infty}\lim e^{-a\left ((x+D(x))^\gamma-x^\gamma\right )}=\\
&=\underset{x\rightarrow +\infty}\lim e^{-ax^\gamma\left (\left (1+\frac{D(x)}x\right )^\gamma-1\right )}.
\end{aligned}$$
If $\underset{x\rightarrow +\infty}\lim \frac{D_i(x)}{x}=\beta\in (0,1]$, then 
$ \lambda^{AMO}_L=0$.
\\
If $\underset{x\rightarrow +\infty}\lim \frac{D_i(x)}{x}=0$, then
$$\lambda^{AMO}_L=\underset{x\rightarrow +\infty}\lim e^{-ax^\gamma\left (\left (1+\frac{D(x)}x\right )^\gamma-1\right )}=\underset{x\rightarrow +\infty}\lim e^{-a\gamma  \frac{D(x)}{x^{1-\gamma}}}$$
and, if $\gamma \geq 1$, then
$ \lambda^{AMO}_L=0$
while, if $\gamma <1$, then
$$\lambda^{AMO}_L=e^{-a\gamma \underset{x\rightarrow +\infty}\lim \frac{D_i(x)}{x^{1-\gamma}}}.$$

\begin{example}\label{infinito}
Assume that there exists 
$\bar x>0$ such that $(\bar x,+\infty)\subset \mathcal A_i$ with 
$\underset{x\rightarrow +\infty}\lim \frac{D_i(x)}{x}=\beta _i$.
We have
\begin{itemize}
 \item 
in the Clayton case $G(x)=(x+1)^{-1/\theta}$, $\theta >0$, $\lambda^{AMO}_L=\left (1+\beta _i\right )^{-1/\theta}\geq 2^{-1/\theta}=\lambda^G_L$;
\item in the Gumbel case $G(x)=e^{-x^{1/\theta}}$, $\theta\geq 1$, with $\beta _1\in (0,1]$, $\lambda^{AMO}_L=0=\lambda^G_L
$;
\item in the Frank case, $G(x)=-\frac 1\theta\ln\left (1+e^{-x}(e^{-\theta}-1)\right )$, $\theta\in\mathbb R$, $\lambda^{AMO}_L=0=\lambda^G_L$.
\end{itemize}
\end{example}

\subsubsection{$\lambda ^{AMO}_U$}
Similarly as done for $\lambda _L^{AMO}$ we can calculate $\lambda _U^{AMO}$ in some specific cases.\medskip

If there exists $\hat x>0$ such that $(0,\hat x)\subset \mathcal A_i$ then 
$$ \lambda^{AMO}_U=2-\underset{x\rightarrow 0}\lim\frac{1-G(x+D_i(x))}{1-G(x)}.$$
If 
$$1-G(x)\underset{x\downarrow 0}\sim cx^\gamma$$
for some $c,\gamma >0$, then 
$$\underset{x\rightarrow 0}\lim\frac{1-G(x+D_i(x))}{1-G(x)}=\underset{x\rightarrow 0}\lim\left (1+\frac{D_i(x)}x\right )^\gamma$$
Hence, if $\underset{x\rightarrow 0}\lim\frac{D_i(x)}x=\beta_i\in[0,1]$, then
$$\lambda^{AMO}_U=2-(1+\beta _i)^\gamma.$$
\bigskip

\begin{example}\label{zero}
Assume that there exists 
$\bar x>0$ such that $(0,\hat x)\subset \mathcal A_i$ with $\underset{x\rightarrow 0^+}\lim \frac{D_i(x)}{x}=
\beta _i\in [0,1]$. We have
\begin{itemize}
 \item 
in the Clayton case $G(x)=(x+1)^{-1/\theta}$, $\theta >0$, $\lambda^{AMO}_U=1-\beta _i\geq 0=\lambda ^G_U$;
\item in the Gumbel case $G(x)=e^{-x^{1/\theta}}$, $\theta\geq 1$, $\lambda^{AMO}_U=2-\left (1+\beta _i\right )^{1/\theta}\geq 2-2^{1/\theta}=\lambda ^G_U$;
\item in the Frank case, $G(x)=-\frac 1\theta\ln\left (1+e^{-x}(e^{-\theta}-1)\right )$, $\theta\neq 0$,
$\lambda^{AMO}_U=1-\beta _i\geq 0=\lambda ^G_U$.
\end{itemize}
\end{example}
\section*{Aknowledgements}
The author thanks Fabrizio Durante and Matthias Scherer for their helpful comments and suggestions.

\end{document}